\documentclass[conference,10pt]{IEEEtran}
\usepackage{amsmath,amsfonts}
\usepackage{algorithmic}
\usepackage{algorithm}
\usepackage{array}
\usepackage{textcomp}
\usepackage{stfloats}
\usepackage{url}
\usepackage{verbatim}
\usepackage{graphicx}
\usepackage{cite}
\usepackage[nolist,nohyperlinks]{acronym}
\begin{acronym}
    \acro{AP}{access point}
    \acro{CCDF}{complementary cumulative density function}
    \acro{CDF}{cumulative density function}
    \acro{PDF}{probability density function}
    \acro{PPP}{Poisson point process}
    \acro{BS}{base station}
    \acro{LOS}{line of sight}
    \acro{SINR}{signal-to-interference-plus-noise ratio}
    \acro{SNR}{signal to noise ratio}
    \acro{mm-wave}{millimeter wave}
    \acro{PV}{Poisson-Voronoi}
    \acro{LOS}{line-of-sight}
    \acro{NLOS}{non line-of-sight}
    \acro{PGFL}{probability generating functional}
    \acro{PLCP}{Poisson line Cox process}
    \acro{MBS}{macro base station}
    \acro{PLT}{Poisson line tessellation}
    \acro{SBS}{small cell base station}
    \acro{RAT}{radio access technique}
    \acro{RCS}{radar cross section}
    \acro{5G}{fifth generation}
    \acro{PLP}{Poisson line process}
    \acro{UE}{user equipment}
\end{acronym}

\usepackage{subfig}
\usepackage{xcolor}
\usepackage{float}
\usepackage{svg}
\usepackage{pdfpages}
\usepackage{comment}

\usepackage{graphicx,wrapfig,lipsum}
\usepackage{rotating}
\usepackage{amsmath, amssymb, amsthm}
\usepackage{stmaryrd}
\usepackage{tikz}
\usepackage{verbatim}
\usepackage{url}
\usepackage[utf8]{inputenc}
\usepackage[english]{babel}
\newtheorem{theorem}{Theorem}[]

\newtheorem{lemma}[]{Lemma}

\usepackage{multicol}

\usepackage{scalerel}
\usepackage{multicol}
\usepackage{setspace}
\newtheorem{definition}{Definition}
\usepackage{booktabs}
\usepackage{bbm}
\usepackage[normalem]{ulem}
\usepackage{color}
\usepackage{dsfont}
\usepackage{bm}
\usepackage{setspace}
\usetikzlibrary{automata}
\usetikzlibrary{arrows}
\usetikzlibrary{shapes.geometric, arrows}
\usepackage[]{footmisc}
\usetikzlibrary{positioning}
\usetikzlibrary{calc}
\usetikzlibrary{arrows}
\allowdisplaybreaks

\usepackage{lipsum}

\usepackage[margin = 0.3cm]{caption}
\usepackage{mathtools}
\usepackage{csquotes}

\makeatletter
\newcommand{\vast}{\bBigg@{3}}
\newcommand{\Vast}{\bBigg@{4}}
\makeatother

\begin{document}

\title{Impact of Urban Street Geometry on the Detection Probability of Automotive Radars}

% \author{\IEEEauthorblockN{Mohammad~Taha~Shah}, {\it Graduate Student Member, IEEE}, \IEEEauthorblockN{Ankit~Kumar}, \IEEEauthorblockN{Gourab~Ghatak}, {\it Member, IEEE}, and \IEEEauthorblockN{Shobha~Sundar~Ram}, {\it Senior Member, IEEE}}
% \thanks{{Mohammad Taha Shah is with the Bharti School of Telecommunication Technology and Management, IIT Delhi (email: bsz218183@iitd.ac.in). Gourab Ghatak is with the Department of Electrical Engineering, IIT Delhi, India (email: gghatak@ee.iitd.ac.in).}

% \author{IEEE Publication Technology,~\IEEEmembership{Staff,~IEEE,}}

\author{
\IEEEauthorblockN{\small Mohammad Taha Shah, Ankit Kumar, \& Gourab Ghatak}
\IEEEauthorblockA{
\textit{\small Indian Institute of Technology Delhi}\\
\small New Delhi, India\\
\small \{bsz218183, bsy227531, gghatak\}@iitd.ac.in}
% \and
% \IEEEauthorblockN{3\textsuperscript{rd} Gourab Ghatak}
% \IEEEauthorblockA{
% \textit{Indian Institute Of Technology Delhi}\\
% New Delhi, India\\
% gourab.ghatak@iitd.ac.in}
\and
\IEEEauthorblockN{\small Shobha Sundar Ram}
\IEEEauthorblockA{
\textit{\small Indraprastha Institute of Information Technology Delhi}\\
\small New Delhi, India\\
\small shobha@iiitd.ac.in}
}
\normalsize
\maketitle

\begin{abstract}
Prior works have analyzed the performance of millimeter wave automotive radars in the presence of diverse clutter and interference scenarios using stochastic geometry tools instead of more time-consuming measurement studies or system-level simulations. In these works, the distributions of radars or discrete clutter scatterers were modeled as Poisson point processes in the Euclidean space. However, since most automotive radars are likely to be mounted on vehicles and road infrastructure, road geometries are an important factor that must be considered. Instead of considering each road geometry as an individual case for study, in this work, we model each case as a specific instance of an underlying Poisson line process and further model the distribution of vehicles on the road as a Poisson point process - forming a Poisson line Cox process. 
Then, through the use of stochastic geometry tools, we estimate the average number of interfering radars for specific road and vehicular densities and the effect of radar parameters such as noise and beamwidth on the radar detection metrics. The numerical results are validated with Monte Carlo simulations. %the significance of channel noise, interfering vehicular radar, and their beamwidth in influencing the probability of achieving effective radar detection. Through our analysis, we can determine the minimum threshold for successful detections and identify the optimal beamwidth that maximises this metric. This valuable information enhances our understanding of the efficient design principles for vehicle radar systems. The proposed technique for improving the beamwidth has the potential to enhance both the safety and technological aspects of automobiles.
\end{abstract}

\begin{IEEEkeywords}  
Stochastic geometry, Poisson line processes, Poisson line Cox process, Poisson point process, automotive radar.
\end{IEEEkeywords}
\section{Introduction}
There has been significant recent interest in using stochastic geometry tools to model and analyze the performance of automotive radar systems ~\cite{al2017stochastic, munari2018stochastic, ren2018performance, park2018analysis, fang2020stochastic,ram2020estimating,ram2021optimization,ram2022estimation,ram2022optimization,singhal2023leo}. Millimeter wave radio channels are extremely complex especially in high mobility environments since they are characterized by significant clutter with short coherence times and interference arising from other radars. Further, there is considerable spatial randomness in the positions and densities of discrete clutter scatterers and interference sources, resulting in significant variations in the system's performance based on the position of the radar, the target scattering response, and the channel conditions. Instead of treating each radar deployment as a separate problem, the prior works have approximated the distribution of the discrete clutter scatterers and interference sources as point processes and utilized the stochastic geometry-based mathematical framework to draw system-level insights regarding radar performance. The main advantage of this approach is that the methodology avoids the requirements of laborious system-level simulations or measurement data collection. 
For example, the authors in \cite{al2017stochastic} modeled the distribution of vehicular automotive radars as a Poisson point process to derive the mean \ac{SINR} from which the radar detection metrics are estimated. Munari {\it et al.}~\cite{munari2018stochastic} used the strongest interferer approximation to determine the radar detection range and false alarm rate, while the radar detection probability was investigated in~\cite{fang2020stochastic}. 
In \cite{ram2020estimating}, the authors modeled the discrete clutter scatterers encountered in monostatic radar scenarios as a Poisson point process and 
estimated the radar detection metrics. They further utilized this framework in \cite{ram2021optimization} to optimize pulse radar parameters for maximizing the probability of detection. They investigated the bistatic radar framework for terrestrial \cite{ram2022estimation} and space-based scenarios in \cite{singhal2023leo} and further extended the model to optimize the time resource management between radar and communications in integrated sensing and communication systems \cite{ram2022optimization}.

In all prior works, the spatial distribution of the radars has been assumed to be entirely random in the Cartesian space. However, in real-world conditions, automotive radars typically mounted on vehicles are more likely to be randomly distributed on roads than off-road sites. Hence, in vehicular networks, an additional aspect of the network geometry that needs consideration is the structure of the streets in a city, which is largely ignored in stochastic geometry-based works. Further, a street-based characterization of radar performance can provide valuable system insights for vehicular applications involving critical control operations such as lane-change assistance and the dissemination of basic safety messages at intersections. In this paper, we will examine how street geometry impacts the interference characteristics of radar using the concept of line processes. In contrast to point processes, which model the locations of isolated events or objects, line processes are random collection of lines in a two-dimensional Euclidean plane and have been used to describe the spatial distribution of roads, rivers, fiber optic cables, and fault lines in different works. Currently, there is limited work on using line processes for modeling automotive radars. Schipper {\it et al}~\cite{schipper2015simulative} analyzed the interference between two radars while modeling the distribution of vehicles along a roadway based on traffic flow patterns. Recently, in~\cite{ghatak2022radar}, authors presented a fine-grained radar detection analysis, highlighting optimal channel access methods. However, the work primarily focuses on a specific highway scenario, potentially limiting its broader applicability.

In this study, we model the network of streets as a Poisson line process. This assumption is especially useful in modeling the road geometries of dense urban areas. Here, at each time instant, the distribution of streets encountered by an ego radar is spatially diverse. Therefore, each case is treated as an instance of an underlying Poisson line process. Similarly, we model the vehicles on each street as a Poisson point process - which is a collection of random points on each line. The choice of the double stochastic process - termed the \ac{PLCP} - is dictated by its mathematical simplicity, spatial independence, and alignment with previous stochastic geometry models. This methodology allows for easy and tractable vehicular radar performance analysis for random distributions of streets and vehicles in terms of the parameters of the network geometry. The key contributions of our work are as follows: First, leveraging the popular \ac{PLCP} model of stochastic geometry, we derive the statistics 
% of the number and the locations of interfering automotive radars to the typical radar mounted on a vehicle. This requires the characterization of the subset of the \ac{PLCP} on an intersection of two convex cones, which has previously not been studied in the literature.
of the distance between interfering and ego radar, which has previously not been studied in the literature. Second, based on the proposed model, we derive the radar detection probability for the ego radar and study the impact of channel noise, the density of the interfering radars, and beamwidth on the detection metrics. Third, we obtain the lower bound on the number of successful detections and use it to find the optimal beamwidth for which the number of successful detections is maximized. Finally, we determine the optimal beamwidth for which the number of successful detections is maximized for specific intensity of vehicles and maximum unambiguous range. 

The paper is organized as follows. The following section presents the system model for the stochastic geometry framework. Sections \ref{sec:AvgPLCP} and \ref{sec:RadarProb} discuss the key theoretical findings while the numerical results are given in Section \ref{sec:Results} followed by the conclusion in Section \ref{sec:Conclusion}.

\section{System Model}
\label{sec:SysModel}

\begin{figure}[t]
\centering
\subfloat[]
{\includegraphics[trim={0cm 0cm 0cm 1.2cm},clip,width=0.3\textwidth]{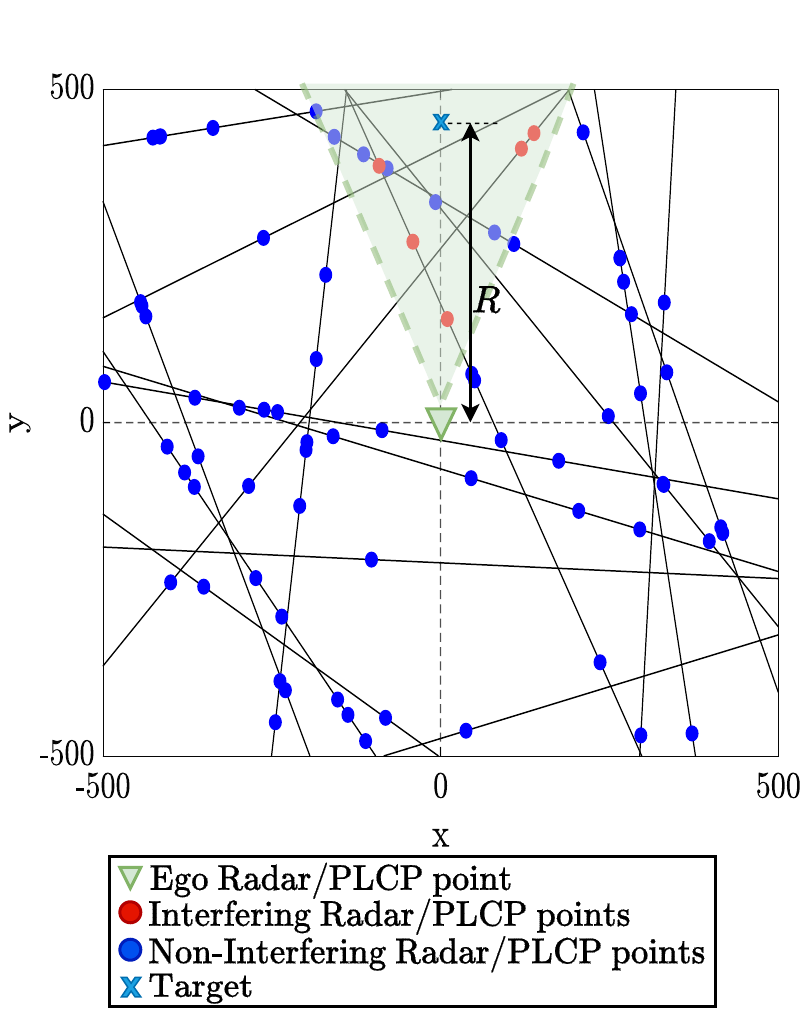}
\label{fig:fig_1}}
\hfill
\subfloat[]
{\includegraphics[trim={4.4cm 4.3cm 4.55cm 1.85cm},clip,width = 0.4\textwidth]{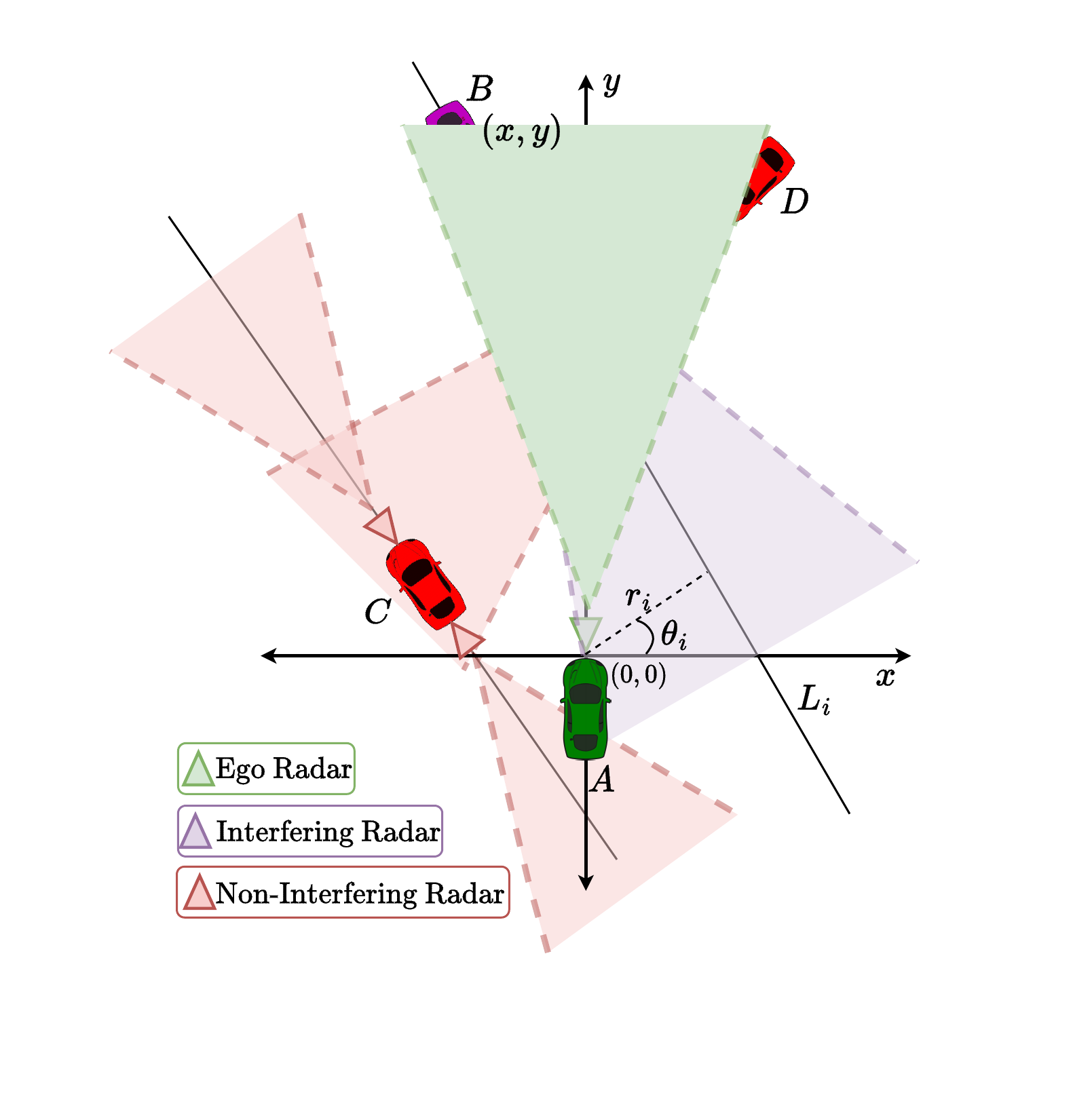}
\label{fig:fig2}}
\caption{(a) Illustration of a scenario showing interfering and non-interfering radar beams w.r.t. ego radar, and (b) A realization of \ac{PLCP} having $\lambda_{\rm L} = 0.005 \,\rm{m}^{-2}$ and $\lambda_{\rm P} = 0.005 \,\rm{m}^{-1}$, with ego radar present at origin.}
\label{fig:result_0_1} 
\end{figure}

\subsection{Network Geometry}
We consider a network of streets modeled as a homogeneous \ac{PLP}, $\mathcal{P} = \{L_1, L_2, \dots\}$, i.e., a stochastic set of lines in the Euclidean plane with the Cartesian directions denoted by $x$ and $y$, respectively. Any line, $L_i$, of $\mathcal{P}$ is uniquely characterized by its distance $r_i$ from the origin and the angle $\theta_i$ between the normal and the $x$-axis. The pair of parameters $(\theta_i,r_i)$ corresponds to a point, $\mathbf{q}_i$, in the representation space $\mathcal{D} \equiv [0,2\pi) \times (0,\infty)$, and is called the generating point of the line $L_i$ in the Euclidean plane. Thus, there is a one-to-one correspondence between $\mathbf{q}_i, i = 1:I$ in $\mathcal{D}$ and the lines, $L_i, i = 1:I$ in $\mathbb{R}^2$. The number of generating points, $I$, in any $S \subset \mathcal{D}$ is driven by the Poisson distribution with parameter $\lambda_{\rm L}|S|$, where $|S|$ represents the Lebesgue measure of $S$. Our analysis focuses on the perspective of an automotive radar mounted on a typical ego vehicle located at the origin of the two-dimensional plane, as shown in Fig.~\ref{fig:fig_1}. The radar has a half-power beamwidth of $\Omega_{\rm B}$, and the target is located at a distance of $R$ from the radar on the same street.

As per the Palm distribution of a \ac{PLP}, we assume that the street with the ego radar and target is represented by $L_1$ with parameters $(\theta_1,r_1)=(0,0)$, while the statistics of the rest of the process remain unaltered~\cite{dhillon2020poisson}. The locations of the vehicles with mounted radars on $L_{i}\textsuperscript{th}$ street follow a one-dimensional \ac{PPP}, $\Phi_{L_i}$, with $\lambda_{\rm P}$ intensity. Note that the \ac{PPP} of vehicles on any street is independent of the corresponding distributions on the other streets. Hence, the complete distribution of the vehicles in the space is a homogeneous \ac{PLCP} $\Phi$, on the domain $\mathcal{P}$ where the PLCP is defined as 
% \begin{align*}
    $\Phi = \bigcup_{L_i \in \mathcal{P}} \Phi_{L_i}$.
% \end{align*}
Therefore, $\Phi$ is a collection of Poisson points on Poisson lines. Fig~\ref{fig:fig_1} shows one realization of the \ac{PLCP}, where the ego radar is present at the origin. The other \ac{PLCP} points representing automotive radars may or may not contribute to the interference. The red dots represent the interfering radars, and the blue dots represent the non-interfering radar.
The ego radar will experience interference if and only if both the ego radar and interfering radar come into each other's radar beams simultaneously. For example, in Fig~\ref{fig:fig2}, we see the radar beams emitted by four vehicles, namely \textit{A}, \textit{B}, \textit{C}, and \textit{D}. The vehicle \textit{A}, located at the origin $\mathbf{0}$, represents the ego radar, and its radar sector is depicted in green. The beams of the interfering and non-interfering vehicles are represented by purple and red, respectively. Therefore, ego radar experiences interference only from \textit{B}, and not from  \textit{C} and \textit{D}. %In the subsequent discussion, we characterize the \textcolor{blue}{interfering distance} and determine its effect on the system performance.

\subsection{Channel Model and SINR}
Let the transmit power be $P$ and the path-loss exponent be $\alpha$. The target is at a distance $R$ and is assumed to have a fluctuating radar cross-section, $\sigma_{\mathbf{c}}$, modeled as an exponential random variable with a mean of $\bar{\sigma}$~\cite{shnidman2003expanded}. Let $G_t$ be the gain of the transmitting antenna and $A_{\rm e}$ the effective area of the receiving antenna aperture. Then the ego radar receives the reflected signal power from the target vehicle with strength
% \begin{align*}
    $S = \gamma\sigma_{\mathbf{c}} P R^{-2\alpha}$
% \end{align*}
where $\gamma = \frac{G_{\rm t}}{(4\pi)^2}A_{\rm e}$. 
Let the coordinates of any \ac{PLCP} point in the Euclidean plane be denoted as, $\mathbf{w} = (x,y)$, such that $x\cos\theta_i + y\sin\theta_i = r_i$ for $(\theta_i,r_i) \in \mathcal{D}$. Due to the millimeter channel effects, the signals undergo multi-path fading, modeled as Rayleigh fading, with a variance of 1. The interference at the ego radar due to any one interfering radar, modeled as one of the \ac{PLCP} points, is then given by,
% \begin{align*}
    $\mathbf{I} = P \gamma h_{\mathbf{w}} ||\mathbf{w}||^{-\alpha}$,
% \end{align*}
where $h_{\mathbf{w}}$ is the fading factor, assuming that all the automotive radars share the same power and gain characteristics. 
% \textcolor{red}{SSR: Are $\gamma_1$ and $\gamma_2$ different?}
The SINR at the ego radar is then given by
\begin{align}
    {\rm SINR} = \frac{\gamma \sigma_{\mathbf{c}} P R^{-2\alpha}}{N + \sum_{\mathbf{w} \in \Phi}4\pi \gamma P h_{\mathbf{w}} ||\mathbf{w}||^{-\alpha}},
\end{align}
where $N_0 = N_{\rm d} W$ is the noise power with $N_{\rm d}$ the noise power density, and $W$ the bandwidth.

\begin{figure}[t]
    \centering
    \includegraphics[trim={7cm 12.9cm 14.5cm 7cm},clip,width = 0.35\textwidth]{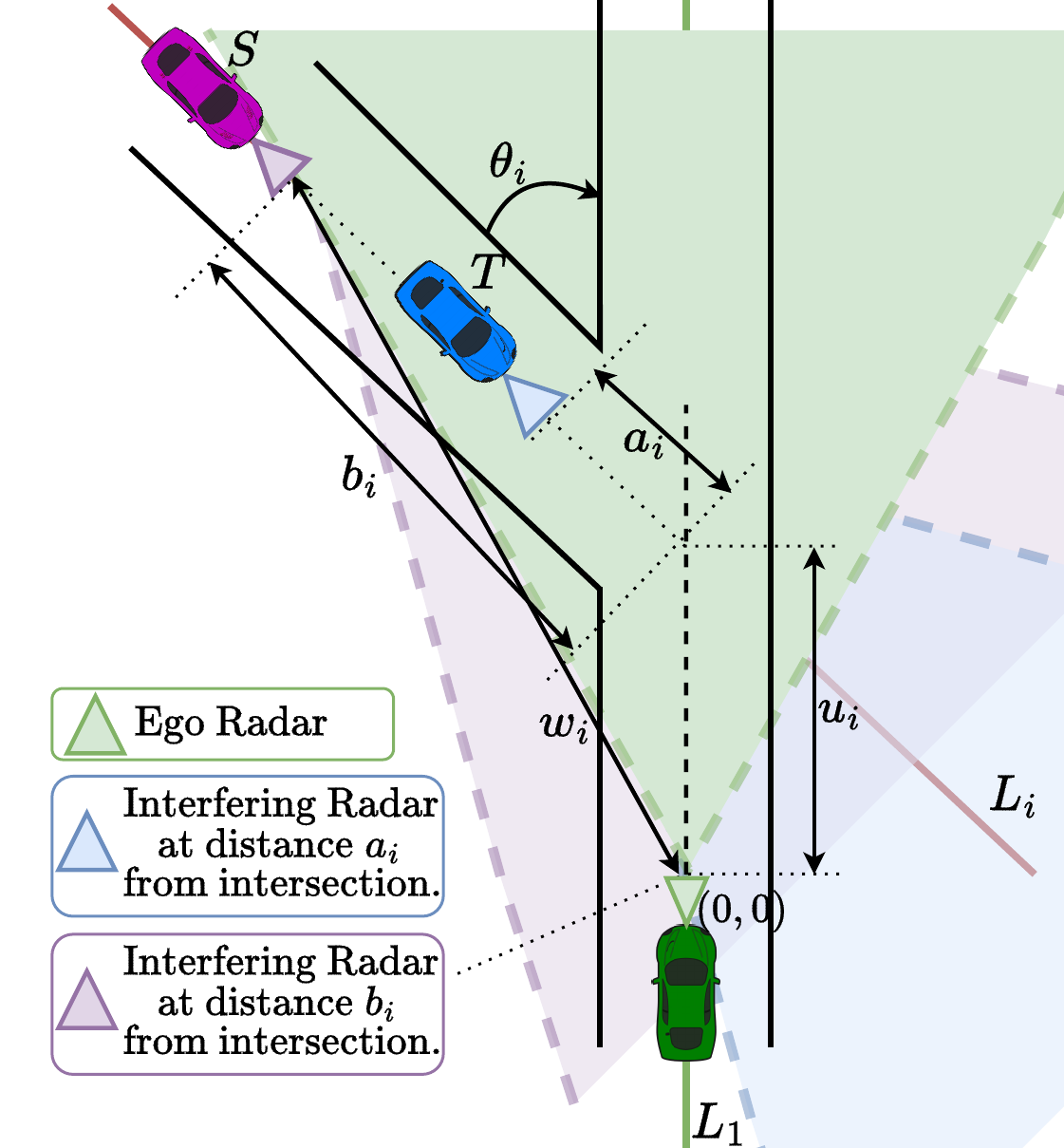}
    \caption{Illustration of a scenario where two radars are present at the edge point of the line $L_i$ inducing interference.}
    \label{fig:fig3}
\end{figure}
\subsection{Interfering Distance}
Our objective is to determine the impact of the interfering radars located on other streets on the detection performance of the ego radar. To do so, \emph{we determine the fraction of $L_i, \forall\, L_i \in \mathcal{P}$ wherein if a radar is present, it will contribute to the interference experienced at the ego radar.} The distance from the ego radar to the intersection of the $L_i\textsuperscript{th}$ line is denoted by $u_i = \frac{r_i}{\sin\theta_i}$, where the angle of intersection is equal to the generating angle $\theta_i$ of the intersecting line. The distance between the interfering radar on $L_i$ from the intersection point is $v_i$ which is referred to as the interfering distance if the other radar and ego radar mutually interfere. Fig~\ref{fig:fig3} illustrates an ego radar with a green radar beam on $L_1$ and radars mounted on other vehicles, \textit{S} and \textit{T}, on $L_i$, with red and purple beams. Specifically, \textit{S} and \textit{T} are at a distance $v_i = b_i$ and $v_i = a_i$ from the intersection point, respectively. These two distances ($a_i$ and $b_i$) represent the bounds of $v_i$ within which another radar will interfere with the ego radar and are determined by the following theorem.
\begin{theorem}
\label{th:theo1}
For any $L_i$ line with parameters $(\theta_i,r_i)$ intersecting line $L_1$ at a distance $u_i$ from the origin, the maximum ($a_i$) and the minimum ($b_i$) distances of the interfering radars from the point of intersection, are given as:
\begin{align}
    a_i \!\!=\!\! 
    \begin{cases}
        -u_i\left(\cos\theta_i - \sin\theta_i\cot\Omega_{\rm B} \right); \hspace*{0cm} \mathrm{for}\; u_i \geq 0, \;\&\; \theta_i \leq \Omega_{\rm B} \\
        \frac{u_i\sin(\theta_i-\Omega_{\rm B})}{\sin\Omega_{\rm B}}; \hspace*{0cm} \mathrm{for}\; u_i \geq 0, \;\&\; \Omega_{\rm B} < \theta_i \leq 2\Omega_{\rm B}\\
        \frac{u_i\sin(\theta_i+\Omega_{\rm B})}{\sin\Omega_{\rm B}}; \hspace*{0cm} \mathrm{for}\; u_i \geq 0, \;\& \\
        \hspace*{3.5cm}\pi-2\Omega_{\rm B} \leq \theta_i \leq \pi - \Omega_{\rm B}\\
        u_i\left(\cos\theta_i - \sin\theta_i\cot\Omega_{\rm B} \right); \hspace*{0cm}\mathrm{for}\; u_i \geq 0, \;\& \\
        \hspace*{3.5cm}\pi-\Omega_{\rm B} \leq \theta_i \leq \pi\\
        \frac{u_i}{\cos\theta_i+\sin\theta_i\cot\Omega_{\rm B}}; \hspace*{0cm} \mathrm{for}\; u_i < 0, \;\&\; \pi < \theta_i \leq \pi + \Omega_{\rm B}\\
        \frac{-u_i}{\cos\theta_i+\sin\theta_i\cot\Omega_{\rm B}}; \hspace*{0cm} \mathrm{for}\; u_i < 0, \;\&\; 2\pi - \Omega_{\rm B} \leq \theta_i <2\pi\\
        0; \hspace*{2.25cm} \mathrm{otherwise}
    \end{cases}
    \label{eq:l_1}
\end{align}
\begin{align}
    b_i \!\!=\!\!
    \begin{cases}
        \infty; \hspace*{0.5cm} \mathrm{for}\; u_i \geq 0, \; \sin\theta_i < \sin\Omega_{\rm B}, \;\&\; \theta_i\leq\pi \\
        \frac{u_i\tan\Omega_{\rm B}}{\sin\theta_i - \cos\theta_i \tan\Omega_{\rm B}};\hspace*{0.07cm} \mathrm{for}\; u_i \geq 0 \;\&\; \Omega_{\rm B} < \theta_i \leq 2\Omega_{\rm B}\\
        \frac{u_i\tan\Omega_{\rm B}}{\sin\theta_i + \cos\theta_i \tan\Omega_{\rm B}}; \hspace*{0.07cm} \mathrm{for}\; u_i \geq 0, \;\& \\
        \hspace*{3.8cm}\pi-2\Omega_{\rm B} \leq \theta_i \leq \pi - \Omega_{\rm B}\\
        \infty; \hspace*{0.5cm} \mathrm{for}\; u_i < 0, \; \sin\theta_i > \sin\Omega_{\rm B}, \;\&\; \pi<\theta_i \leq 2\pi\\
        0; \hspace*{0.7cm} \mathrm{otherwise}
    \end{cases}
    \label{eq:l_2}
\end{align}
\end{theorem}
\begin{proof}
The proof of the above theorem can easily be derived by simple trigonometric manipulations and tricks. To distinguish between various scenarios, it is necessary to ascertain the distance of the point of intersection from the ego radar, as well as the angle of intersection. In Eq~\eqref{eq:l_1} and~\eqref{eq:l_2}, we see that it consists of two major cases when the intersection is happening ahead of ego radar, i.e., $u_i \geq 0$ and the case when the intersection is happening behind the ego radar, i.e., $u_i<0$. For the first case, 4 sub-cases depend on the intersecting angle $\theta_i$. The first two sub-cases are when $\theta_i \leq \frac{\pi}{2}$ and other two are for when $\frac{\pi}{2} < \theta_i \leq \pi$. To better explain the proof, we have divided it into three parts of 3 major scenarios it encounters.

\begin{figure*}[t]
\centering
\subfloat[]
{\includegraphics[trim={1.7cm 3cm 1.2cm 1.3cm},clip,width=0.25\textwidth]{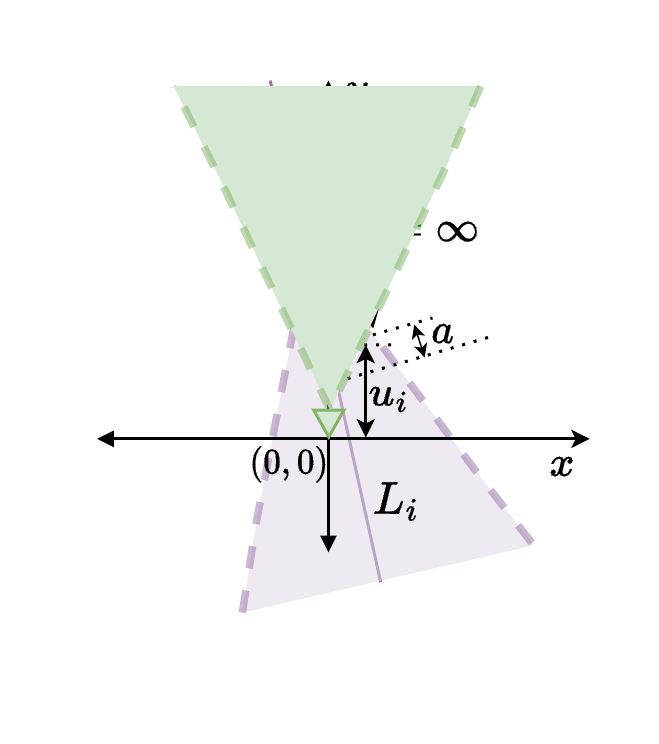}
\label{fig:SC_1}}
\hfill
\subfloat[]
{\includegraphics[trim={0.6cm 2.8cm 1.4cm 1.4cm},clip,width=0.25\textwidth]{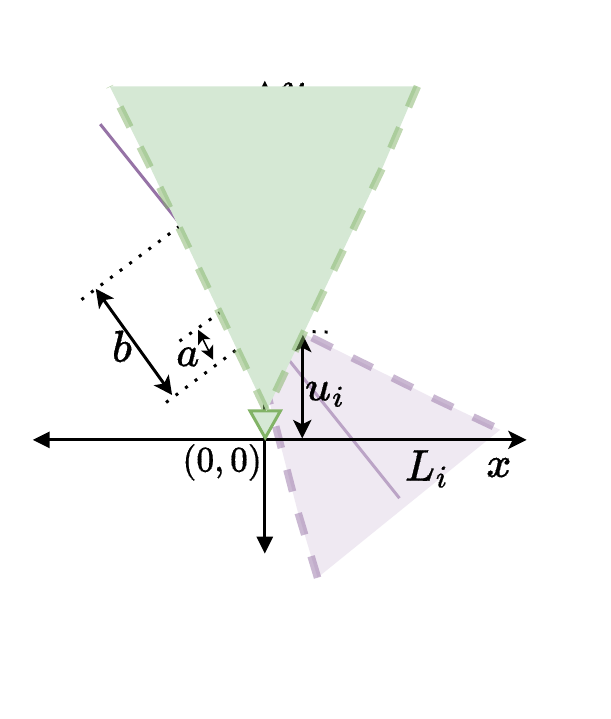}
\label{fig:SC_2}}
\hfill
\subfloat[]
{\includegraphics[trim={2cm 1.1cm 1.4cm 1.3cm},clip,width=0.25\textwidth]{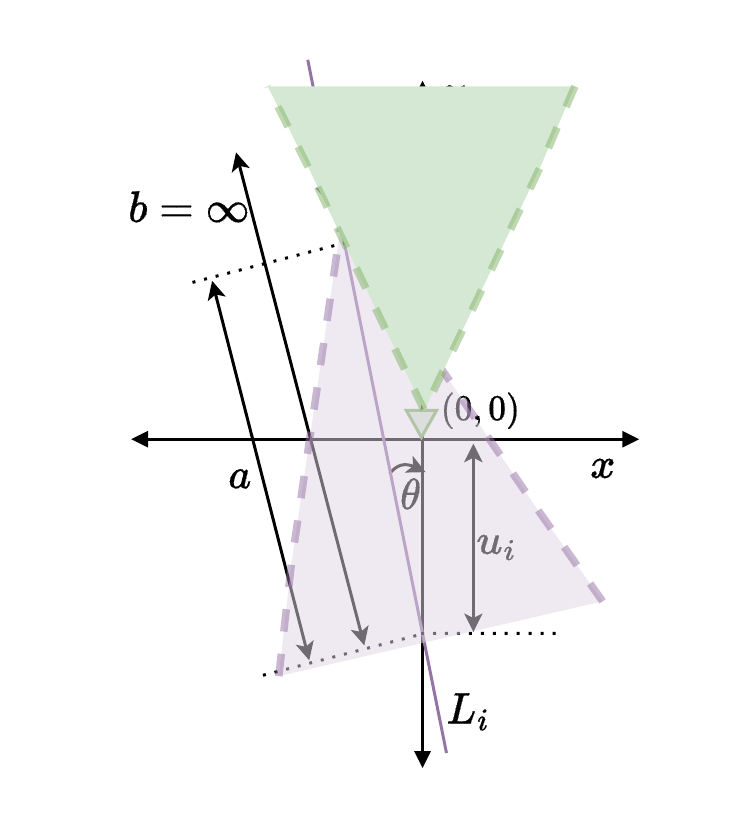}
\label{fig:SC_3}}
\caption{(a) Case 1: When the intersecting line infinitely inside the radar sector, (b) Case 2: When the intersecting line cuts out of the radar sector, and (c) Case 3: When the intersecting line intersects behind the origin.}
\end{figure*}

\textbf{Case 1:} $u_i \geq 0$, and $\theta_i \leq \Omega_{\rm B}$ or $\pi-\Omega_{\rm B} < \theta_i \leq\pi$.\par
In this scenario, we consider the case when the interfering line $L_i$ extends infinitely within the radar sector, as shown in Fig.~\ref{fig:SC_1}.
% \begin{figure}[h]
% \centering
% \includegraphics[trim={1.7cm 3cm 1.2cm 1.3cm},clip,width=0.35\textwidth]{SC1.pdf}
% \caption{Case 1}
% \label{fig:SC_1}
% \end{figure}
If the intersecting angle is less than half beamwidth, line $L_i$ will always remain inside the radar sector, thus the upper limit of interfering distance $v_i$, i.e., $b_i = \infty$. The distance $a_i$ is measured from the point of intersection to a point where both radars create mutual interference. Using simple trigonometric formulas, we get $a_i$ as
\begin{align*}
    a_i = -u_i\left(\cos\theta_i - \sin\theta_i\cot\Omega_{\rm B} \right)
\end{align*} 
for $\theta_i \leq \Omega_{\rm B}$, and for $\pi-\Omega_{\rm B} < \theta_i \leq\pi$
\begin{align*}
     a_i = u_i\left(\cos\theta_i - \sin\theta_i\cot\Omega_{\rm B} \right)
\end{align*}

\textbf{Case 2:} $u_i \geq 0$, and $ \Omega_{\rm B} < \theta_i \leq 2\Omega_{\rm B}$ or $\pi-2\Omega_{\rm B} \leq \theta_i \leq \pi - \Omega_{\rm B}$.\par
In this scenario, the interfering line $L_i$ will cut out of the radar sector, i.e., it will not lie infinitely within the radar sector as shown in Fig~\ref{fig:SC_2}.
% \begin{figure}[h]
% \centering
% \includegraphics[trim={0.6cm 2.8cm 1.4cm 1.4cm},clip,width=0.35\textwidth]{SC2.pdf}
% \caption{Case 2}
% \label{fig:SC_2}
% \end{figure}
This scenario occurs if the intersecting angle is greater than half beamwidth and less than beamwidth. An upper limit exists on the value of $v_i$, because at some distance from the intersection point, interfering radar will not lie in the radar sector. Likewise, $a_i$ is the nearest distance from the intersection point after which interference cannot take place, using trigonometric and basic geometric properties $a_i$ and $b_i$ are given as
\begin{align*}
    a_i = \frac{u_i\sin(\theta_i-\Omega_{\rm B})}{\sin\Omega_{\rm B}}, \quad b_i = \frac{u_i\tan\Omega_{\rm B}}{\sin\theta_i - \cos\theta_i \tan\Omega_{\rm B}}
\end{align*}
for $ \Omega_{\rm B} < \theta_i \leq 2\Omega_{\rm B}$, and for $\pi-2\Omega_{\rm B} \leq \theta_i \leq \pi - \Omega_{\rm B}$
\begin{align*}
    a_i = \frac{u_i\sin(\theta_i+\Omega_{\rm B})}{\sin\Omega_{\rm B}}, \quad b_i = \frac{u_i\tan\Omega_{\rm B}}{\sin\theta_i + \cos\theta_i \tan\Omega_{\rm B}}
\end{align*}

\textbf{Case 3:} $u_i < 0$, and $\pi < \theta_i \leq \pi + \Omega_{\rm B}$ or $2\pi - \Omega_{\rm B} \leq \theta_i <2\pi$. \par
This is the third case where a vehicle can create interference, which arises when the point of intersection is beyond ego radar, i.e., it lies on the negative y-axis.
% \begin{figure}[h]
% \centering
% \includegraphics[trim={2cm 1.1cm 1.4cm 1.3cm},clip,width=0.35\textwidth]{SC3.pdf}
% \caption{Case 3}
% \label{fig:SC_3}
% \end{figure}
Fig~\ref{fig:SC_3} illustrates this scenario when the intersection point lies behind the origin. The intersection point will lie behind the ego radar only if $\theta_i>\pi$, and if its value is present in a certain range. Using basic trigonometric manipulations, we get the value of $a_i$ as
\begin{align*}
    a_i = \frac{u_i}{\cos\theta_i+\sin\theta_i\cot\Omega_{\rm B}}
\end{align*}
for $\pi < \theta_i \leq \pi + \Omega_{\rm B}$, and for $2\pi - \Omega_{\rm B} \leq \theta_i <2\pi$
\begin{align*}
    a_i = \frac{-u_i}{\cos\theta_i+\sin\theta_i\cot\Omega_{\rm B}}
\end{align*}
Likewise, the value of $b_i$ will be infinity because the interfering line will exist infinitely within the radar sector. 

This completes the proof of theorem~\ref{th:theo2}
\end{proof} 
For each sub-case, the distance between the ego and interfering radar is given as
\begin{align}
    w_i=\sqrt{(u_i \pm v_i\cos\theta_i)^2 + (v_i\sin\theta_i)^2}.
    \label{eq:eq_w}
\end{align}
\section{Average number of \ac{PLCP} points in Radar sector}
\label{sec:AvgPLCP}
In this section, we calculate the average number of interfering radars that fall within the beam of the ego radar. This involves characterizing the average number of \ac{PLCP} points, $n(R)$, in the bounded ego radar sector with the maximum unambiguous range $R$ denoted by $\mathcal{C}_{(0,0)}^{+}\!(R)$. 

We identify $\mathcal{C}_{(0,0)}^{+}\!(R)$ as a collection of points $(p, q)$ in the Euclidean plane such that angle made by the vectors $(p, q)$ and $(0,-1)$ is greater than $\cos\Omega_{\rm B}$ and also, the distance between ego radar and $(p, q)$ is less than $R$, as shown below
\small
\begin{align*}
    \mathcal{C}_{(0,0)}^{+}\!(R) \!=\! \bigg\{(p,q) \colon \frac{q}{\sqrt{p^2 + q^2}} > \cos \Omega_{\rm B},\, \& \sqrt{p^2 + q^2} \leq R \bigg\}.
\end{align*}
\normalsize
Next, we note that $n(R)$  depends on the average length of \ac{PLP} lines falling inside $\mathcal{C}_{(0,0)}^{+}\!(R)$. For this, we consider the expectation of generating points over $\mathcal{D}_R \equiv [0,\pi) \times (0,R)$ which is a subset of $\mathcal{D}$ having $r<R$.
\begin{definition}
Thus the average number of \ac{PLCP} points falling inside $\mathcal{C}_{(0,0)}^{+}\!(R)$ is
\begin{align}
    n(R) = \lambda_{\rm P} l_{\rm avg} = \lambda_{\rm P} \mathbb{E}_{\mathcal{D}_{R}} \bigg[\left|L \cap \mathcal{C}_{(0,0)}^{+}(R)\right|_1\bigg]
    \label{eq:n_R}
\end{align}
where $|\cdot|_1$ is the Lebesgue measure in one dimension and $L$ is a line of the \ac{PLP}
\end{definition}
We determine the average length of a single line within $\mathcal{C}_{(0,0)}^{+}\!(R)$, as stated in following Lemma~\ref{lm:lemma_1}.
\begin{lemma}
\label{lm:lemma_1}
The length of a line parameterized by $(\theta,r)$ falling inside the bounded radar sector $\mathcal{C}_{(0,0)}^{+}\!(R)$ of half beamwidth $\Omega_{\rm B}$ and maximum unambiguous range $R$ is,
\begin{align}
    l=
    \begin{cases}
        l_1 \!=\! \frac{u\tan\Omega_{\rm B}}{\sin\theta + \cos\theta\tan\Omega_{\rm B}} + \sqrt{R^2 -u^2\sin^2\theta} - u\cos\theta; \\
        \hspace*{2cm}\mathrm{for}\; 0 \leq u \leq R, \;\mathrm{and}\; 0 < \theta \leq  \alpha_{\rm n}\\
        l_2 \!=\! \frac{2 u \sin\theta \tan\Omega_{\rm B}}{\sin^2\theta - \cos^2\theta\tan^2\Omega_{\rm B}}; 
        \hspace*{0.5cm}\mathrm{for}\; 0 \leq u \leq R\cos\Omega_{\rm B}, \\
        \hspace*{3.5cm}\mathrm{and}\; \alpha_{\rm n} < \theta \leq \pi-\alpha_{\rm n} \\
        l_3 \!=\! \frac{u\tan\Omega_{\rm B}}{\sin\theta - \cos\theta\tan\Omega_{\rm B}} + \sqrt{R^2 -u^2\sin^2\theta} + u\cos\theta;\\
        \hspace*{2cm}\mathrm{for}\; 0 \leq u \leq R \;\;\mathrm{and}\; \pi-\alpha_{\rm n} < \theta \leq \pi \\
        l_4 \!=\! 2\sqrt{R^2 -u^2\sin^2\theta}; \hspace*{0.6cm}\mathrm{for}\;  R\cos\Omega_{\rm B} \leq u \leq R \\
        \hspace*{3.5cm}\mathrm{and}\; \alpha_{\rm n} < \theta \leq \pi -\alpha_{\rm n}\\
        0; \hspace*{1.7cm}\mathrm{otherwise}
    \end{cases}
    \label{eq:lengthofline}
\end{align}
where $\alpha_{\rm n} = \arctan{\left(\frac{R\sin\Omega_{\rm B}}{|R\cos\Omega_{\rm B} - u|}\right)}$.
\end{lemma}
\begin{proof}
\textbf{Case 1:} $0 \leq u \leq R$ any $0 < \theta \leq  \alpha_{\rm n}$.\par
Fig~\ref{fig:al1} illustrates the case when intersecting line $L$ intersects one of the edge lines and circular curvature of the radar sector, and it will happen only when $0 < \theta \leq  \alpha_{\rm n}$. 
\begin{figure}[h]
    \centering
    \includegraphics[width = 0.45\textwidth]{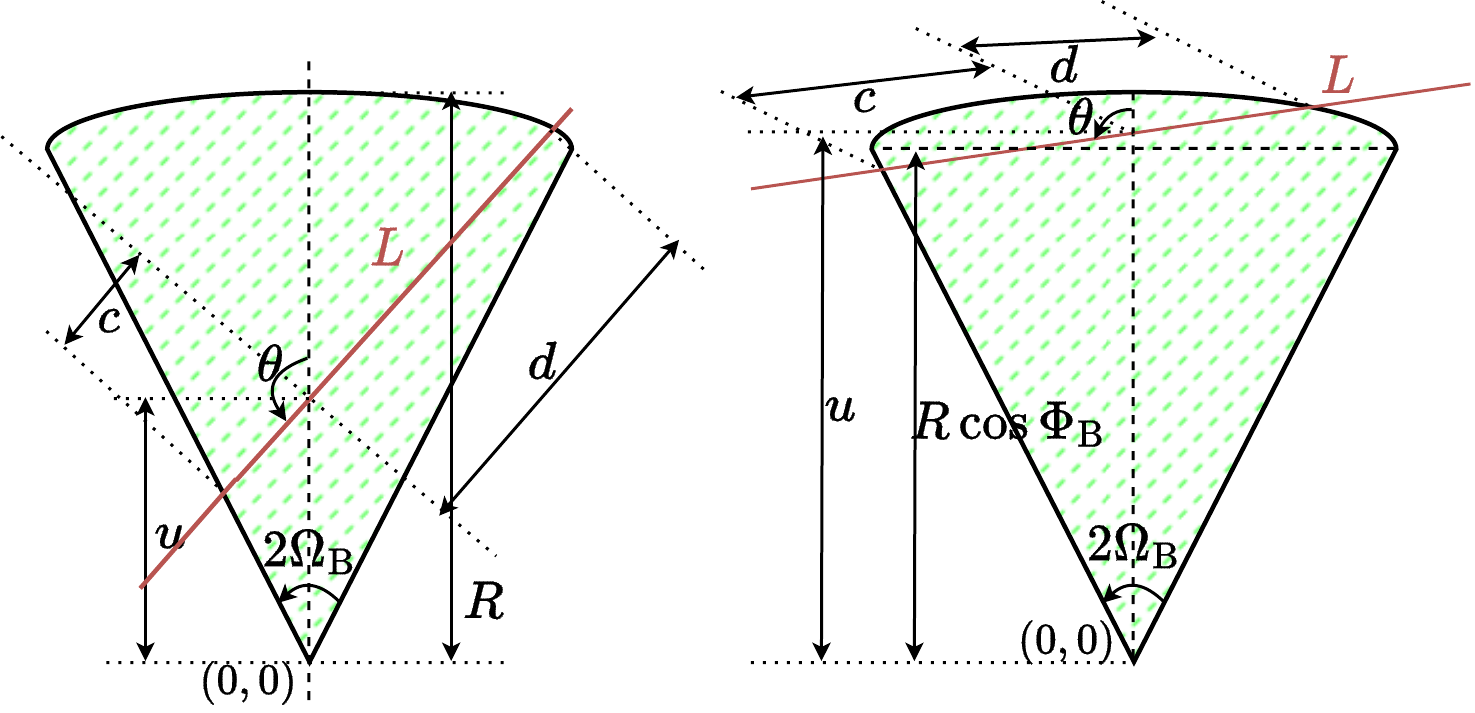}
    \caption{Case 1: When line $L$ is intersecting one of the edges and the circular arc.}
    \label{fig:al1}
\end{figure}
As shown in Fig~\ref{fig:al1}, the value of $l_1$ is composed to two variables $c$ and $d$ i.e., $l_1 = c+d$. we can find $c$ by simple trigonometric operations,
\begin{align*}
    u\tan\Omega_{\rm B} &= c\sin{\theta} + c\cos\theta\tan\Omega_{\rm B} \\
    c &= \frac{u\tan\Omega_{\rm B}}{\sin\theta + \cos\theta\tan\Omega_{\rm B}}
\end{align*}
Likewise, using Pythagoras' theorem, the length of $d$ is,
\begin{align*}
    R^2 &= (d\sin{\theta})^2 + (d\cos{\theta} + u)^2 \\
    d &= \sqrt{R^2 -u^2\sin^2\theta} - u\cos\theta
\end{align*}
Therefore total length of line $l_1$ in this case is as,
\begin{align}
    l_1 &= \frac{u\tan\Omega_{\rm B}}{\sin\theta + \cos\theta\tan\Omega_{\rm B}} + \sqrt{R^2 -u^2\sin^2\theta} - u\cos\theta
    \label{eq:case_1}
\end{align}

\textbf{Case 2:} $0 \leq u \leq R\cos\Omega_{\rm B}$, $\alpha_{\rm n} < \theta \leq \pi-\alpha_{\rm n}$.\par 
Fig~\ref{fig:al2} depicts the scenario in which the intersection distance $u$ is smaller than the $R\cos\Omega_{\rm B}$ and $\alpha_{\rm n} < \theta \leq \pi-\alpha_{\rm n}$.
\begin{figure}[t]
    \centering
    \includegraphics[width = 0.2\textwidth]{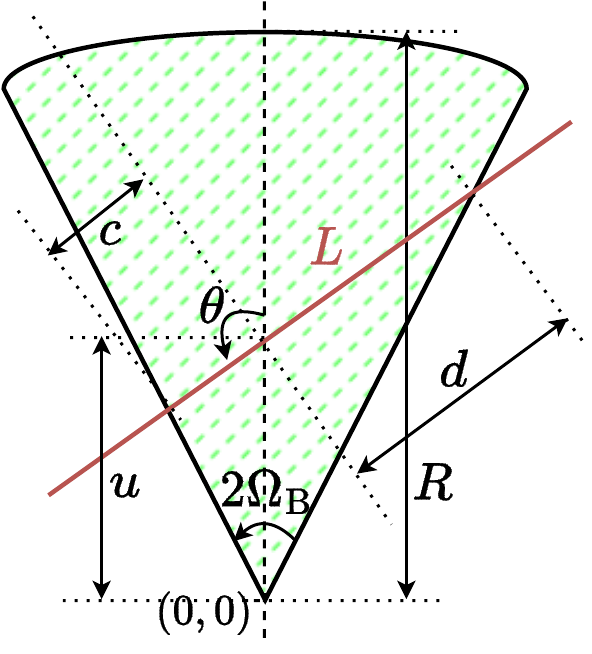}
    \caption{Case 2: When line $L$ is intersecting both one of the edges and not the circular arc.}
    \label{fig:al2}
\end{figure}
In this scenario, line $L$ intersects only the edge lines of the radar sector rather than intersecting the circular arc. The equation of $c$ and $d$ is as,
\begin{align*}
    c &= \frac{u\tan\Omega_{\rm B}}{\sin{\theta} + \cos{\theta}\tan\Omega_{\rm B}}
\end{align*}
and $d$ can be found as,
\begin{align*}
    \tan\Omega_{\rm B} &= \frac{d\sin{\theta}}{u + d\cos{\theta}} \\
    \therefore\, d &= \frac{u\tan\Omega_{\rm B}}{\sin{\theta} - \cos{\theta}\tan\Omega_{\rm B}}
\end{align*}
Therefore total length of line $l_2$ in this case is as,
\begin{align}
    l_2 &= u\tan\Omega_{\rm B}\frac{2\sin\theta}{\sin^2\theta - \cos^2\theta\tan^2\Omega_{\rm B}}
    \label{eq:case_2}
\end{align}

\textbf{Case 3:} $0 \leq u \leq R$, and $\pi-\alpha_{\rm n} < \theta \leq \pi$\par
This case is the same as Case 1, but the only difference is $\theta>\frac{\pi}{2}$. The line $L$ will intersect only with the edge line and the circular arc. The value of $c$ and $d$ are similar as of~\eqref{eq:case_1}, only $\theta$ is replaced with $\pi - \theta$. Thus,
\begin{align}
    l_3 &= \frac{u\tan\Omega_{\rm B}}{\sin\theta - \cos\theta\tan\Omega_{\rm B}} + \sqrt{R^2 -u^2\sin^2\theta} + u\cos\theta
    \label{eq:case_3}
\end{align}

\textbf{Case 4:} $R\cos\Omega_{\rm B} \leq u \leq R$, and $\alpha_{\rm n} < \theta \leq \pi -\alpha_{\rm n}$.\par
In this case, the line $L$ intersects only circular arc region at two places; this case can only arise when the interesting line distance $u$ is greater than $R\cos\Omega_{\rm B}$ and $\alpha_{\rm n} < \theta \leq \pi -\alpha_{\rm n}$. 
\begin{figure}[t]
    \centering
    \includegraphics[width = 0.25\textwidth]{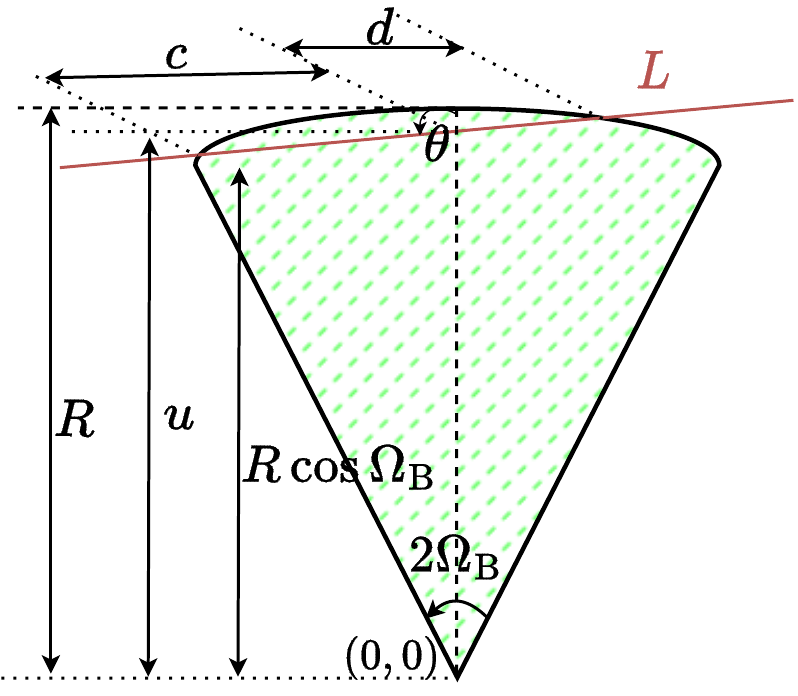}
    \caption{Case 4: When line $L$ intersects only the circular arcs, not the edges.}
    \label{fig:al4}
\end{figure}
Here, the total length $l_4$ can be found by first finding $c$ and $d$. Both the values of $c$ and $d$ are determined using Pythagoras theorem and basic trigonometric operations as,
\begin{align*}
    R^2 &= (c\sin{\theta})^2 + (u - c\cos{\theta})^2 \\
    % R^2 &= c^2 - 2cu\cos{\theta} + u^2(\cos^2{\theta} + \sin^2{\theta})\nonumber\\
    \therefore\, c &= \sqrt{R^2 -u^2\sin^2\theta} + u\cos\theta
\end{align*}
and
\begin{align*}
    R^2 &= (d\sin{\theta})^2 + (d\cos{\theta} + u)^2\\
    % R^2 &= d^2 + 2du\cos{\theta} + u^2(\cos^2{\theta} + \sin^2{\theta})\nonumber\\
    \therefore\, d &= \sqrt{R^2 -u^2\sin^2\theta} - u\cos\theta
\end{align*}
Therefore,
\begin{align}
    l_4 &= 2\sqrt{R^2 -u^2\sin^2\theta}.
    \label{eq:case_4}
\end{align}
This completes the proof of theorem~\ref{th:theo1}.
\end{proof}
Following Lemma~\ref{lm:lemma_1}, we derive the average length of $\ac{PLP}$ lines falling inside $\mathcal{C}_{(0,0)}^{+}\!(R)$ in Theorem~\ref{th:theo2}.
\begin{theorem}
\label{th:theo2}
In a \ac{PLP}, the average length of line segments falling inside the bounded radar sector $\mathcal{C}_{(0,0)}^{+}\!(R)$  is
\begin{align*}
    l_{\rm avg} = \mathbb{E}_{\mathcal{D}_{R}} \bigg[\left|L \cap \mathcal{C}_{(0,0)}^{+}(R)\right|_1\bigg] = 2\pi\lambda_{\rm L}R\bar{l} 
\end{align*}
where
\begin{align*}
    \bar{l} &= \frac{1}{\pi R} \bigg[\int_0^{R}\!\! \int_0^{\alpha_{\rm n}} l_1\, {\rm d}\theta\, {\rm d}u + \int_0^{R\cos\Omega_{\rm B}}\!\! \int_{\alpha_{\rm n}}^{\pi-\alpha_{\rm n}} l_2\, {\rm d}\theta\, {\rm d}u \\
    &\hspace*{1cm} + \int_{0}^R\!\! \int_{\pi-\alpha_{\rm n}}^{\pi} l_3\, {\rm d}\theta\, {\rm d}u + \int_{R\cos\Omega_{\rm B}}^R\!\! \int_{\alpha_{\rm n}}^{\pi-\alpha_{\rm n}} l_4\, {\rm d}\theta\, {\rm d}u \bigg]
\end{align*}
\end{theorem}
\begin{proof}
The proof follows by first determining the length of a single line with parameter $(\theta,r)$ inside $\mathcal{C}_{(0,0)}^{+}\!(R)$. In order to determine the average length of line segments inside the conic section, we use Eq.~\eqref{eq:case_1}, \eqref{eq:case_2}, \eqref{eq:case_3} and \eqref{eq:case_4}. Thus the average length of single line is given by $\bar{l}$ as
\begin{align*}
    \bar{l} &= \frac{1}{\pi R} \bigg[\int_0^{R}\!\! \int_0^{\alpha_{\rm n}} l_1\, {\rm d}\theta\, {\rm d}u + \int_0^{R\cos\Omega_{\rm B}}\!\! \int_{\alpha_{\rm n}}^{\pi-\alpha_{\rm n}} l_2\, {\rm d}\theta\, {\rm d}u \\
    &\hspace*{1cm} + \int_{0}^R\!\! \int_{\pi-\alpha_{\rm n}}^{\pi} l_3\, {\rm d}\theta\, {\rm d}u + \int_{R\cos\Omega_{\rm B}}^R\!\! \int_{\alpha_{\rm n}}^{\pi-\alpha_{\rm n}} l_4\, {\rm d}\theta\, {\rm d}u \bigg]
\end{align*}
By integrating $l_1$, $l_2$, $l_3$, and $l_4$ from the limits in which they exist, we divide by $\pi R$, because lines having $r>R$ can never lie inside $\mathcal{C}_{(0,0)}^{+}\!(R)$. The total average length of all the lines, when the generating points are poisson distributed in $\mathcal{D}_R$ can be found by using Campbell's theorem. This gives $l_{\rm avg}$ and completes the proof.   
\end{proof}
\section{Detection Success Probability}
\label{sec:RadarProb}
The detection \textit{success probability} at a threshold $\beta$ is defined as the \ac{CCDF} of SINR,
% \begin{align}
    $p_{\rm D}(\beta) = \mathbb{P}[\rm{SINR} > \beta]$.
% \end{align}
This represents the probability that an attempted detection by the ego radar of the target located at a distance $R$ is successful. %In what follows, we refer to the detection success probability as the \textit{success probability}.
\begin{theorem}
\label{th:theo3}
For the network where locations of the vehicles are modeled as \ac{PLCP} $\Phi$, the detection success probability for an ego radar is
\begin{align}
    &p_{\rm D}(\beta) = e(R)\, \exp\Bigg[-\lambda_{\rm L} \int_0^{2\pi} \int_{\mathbb{R}^{+}} 1 - \\
    &\hspace*{0.5cm}\exp{\left(-\lambda_{\rm L} \int_a^b 1 - \frac{1}{1+\beta^\prime ||\mathbf{w}||^{-\alpha}} \,{\rm d}v\right)}{\rm d}\theta\, {\rm d}r\Bigg]
    \label{eq:main_eq}
\end{align}
where $\beta^\prime = \frac{4\pi\beta}{\bar{\sigma} R^{-2\alpha}}$, $e(R) = \exp{\left(\frac{-\beta N}{\bar{\sigma}\gamma PR^{-\alpha}}\right)}$, and $||\mathbf{w}||$ is given in~\eqref{eq:eq_w}.
\begin{proof}
From the definition of success probability, we have
\begin{align*}
    &p_{\rm D}(\beta) = \mathbb{P}\left(\frac{\gamma \sigma_{\mathbf{c}}PR^{-2\alpha}}{N + \sum_{\mathbf{w} \in \Phi} 4\pi \gamma P h_{\mathbf{w}} ||\mathbf{w}||^{-\alpha}} > \beta\right) \\
    &= \mathbb{P} \left(\sigma_{\mathbf{c}} > \frac{\beta(N + \sum_{\mathbf{w} \in \Phi} 4\pi \gamma P h_{\mathbf{w}} ||\mathbf{w}||^{-\alpha})}{\gamma PR^{-2\alpha}} \right) \\
    &\overset{(a)}{=} \mathbb{E}_{\Phi, h_{\mathbf{w}}} \left[\exp{\left(-\frac{\beta(N + \sum_{\mathbf{w} \in \Phi} 4\pi \gamma P h_{\mathbf{w}} ||\mathbf{w}||^{-\alpha})}{\bar{\sigma}\gamma PR^{-2\alpha}}\right)}\right] \\
    &= \exp{\left(\frac{-\beta N}{\bar{\sigma}\gamma PR^{-\alpha}}\right)} \\
    &\hspace*{1cm}\mathbb{E}_{\Phi, h_{\mathbf{w}}} \left[\exp{\left(-\frac{\beta \sum_{\mathbf{w} \in \Phi_{\rm L}} 4\pi \gamma P h_{\mathbf{w}} ||\mathbf{w}||^{-\alpha}}{\bar{\sigma}\gamma PR^{-2\alpha}}\right)}\right] \\
    &= e(R)\, \mathbb{E}_{\mathcal{D}} \vast(\prod_{(\theta,r) \in \mathcal{D}}\; \mathbb{E}_{\Phi_{\rm L}} \Bigg( \prod_{\mathbf{w} \in \Phi_{\rm L}} \\
    &\hspace*{2cm}\mathbb{E}_{h_{\mathbf{w}}} \left[\exp{\left(-\frac{4\pi \beta P h_{\mathbf{w}} ||\mathbf{w}||^{-\alpha}}{\bar{\sigma} PR^{-2\alpha}}\right)}\right] \Bigg)\vast) \\
    &\overset{(b)}{=} e(R)\, \mathbb{E}_{\mathcal{D}} \vast(\prod_{(\theta,r) \in \mathcal{D}}\; \mathbb{E}_{\Phi_{\rm L}} \Bigg( \prod_{\mathbf{w} \in \Phi_{\rm L}} \frac{1}{1+\beta^\prime ||\mathbf{w}||^{-\alpha}} \Bigg)\vast) \\
    &\overset{(c)}{=} e(R)\, \exp\Bigg[-\lambda_{\rm L} \int_0^{2\pi} \int_{\mathbb{R}^{+}} 1 - \\
    &\hspace*{0.8cm}\exp{\left(-\lambda_{\rm P} \int_a^b 1 - \frac{1}{1+\beta^\prime ||\mathbf{w}||^{-\alpha}} \,{\rm d}v\right)}{\rm d}\theta\, {\rm d}r\Bigg] \\
\end{align*}
Step (a) follows from the exponential distribution of $\sigma_{\mathbf{c}}$. Step (b) follows by taking the average over fading $h_{\mathbf{w}}$. Step (c) follows from the Laplace functional~\cite{dhillon2020poisson} of \ac{PLCP}.
\end{proof}
\end{theorem}
% We see that the success probability is a function of $\lambda_{\rm L}$, the intensity of streets, $\lambda_{\rm P}$, the of vehicles per meter, and the target distance $R$. \textcolor{red}{This sentence is non-informative - should be removed and replaced with more informative sentences.}
In section~\ref{sec:section_5}, we will see the performance analysis of success probability, $p_{\rm D}(\beta)$ for varying values of beamwidth, ranging distance, and intensity of vehicles.

\subsection{Number of Successful Target Detections}
Next, we determine the lower bound on the average number of successful detections, $n_{\rm D}$, which is the product of the probability of successful detection and the average number of \ac{PLCP} points present within the bounded radar sector created by the radar beam of ranging distance $R$,
\begin{align*}
    n_{\rm D} \geq n(R) \times p_{\rm D} = l_{\rm avg} \lambda_{\rm P} p_{\rm D}.
\end{align*}
We see that $n_{\rm D}$ is a lower bound because we are only considering targets which are at ranging distance less than $R$. 
Let us take into consideration three targets, one at a distance less than R, one at a distance greater than R, and one located precisely at R. Let their success probability of their detection be $p_{1}$, $p_{2}$, and $p_{\rm D}$ respectively. From~\eqref{eq:main_eq} we see that success probability decreases as $R$ increase, thus $p_1\!>$$p_{\rm D}\!>$$p_2$. Now the number of successful detection is $n_{\rm D} = p_1 + p_{\rm D} + p_2$. As we are considering targets only upto distance $R$, and $p_2 > 0$, therefore we can lower bound our number of successful detections as $n_{\rm D} \geq p_1 + p_{\rm D}$, which is further bounded to $n_{\rm D} \geq 2p_{\rm D}$, because $p_1 > p_{\rm D}$. In our case the average number of potential targets are $n(R)$ and $p_{\rm D}$ is the detection probability at distance $R$. Thus we lower bound the average number of successful detection as $ n_{\rm D} \geq n(R) \times p_{\rm D}$. 

$n_D$ is a reliable measure, offering an approximation of the minimum number of successful detections. Therefore, it serves as a highly significant instrument for the optimization of radar system design, 
% the improvement of dependability, 
and the facilitation of decision-making processes in the advancement and implementation of automotive radar technologies. 
\begin{figure*}[t]
\centering
\subfloat[]
{\includegraphics[width=0.32\textwidth]{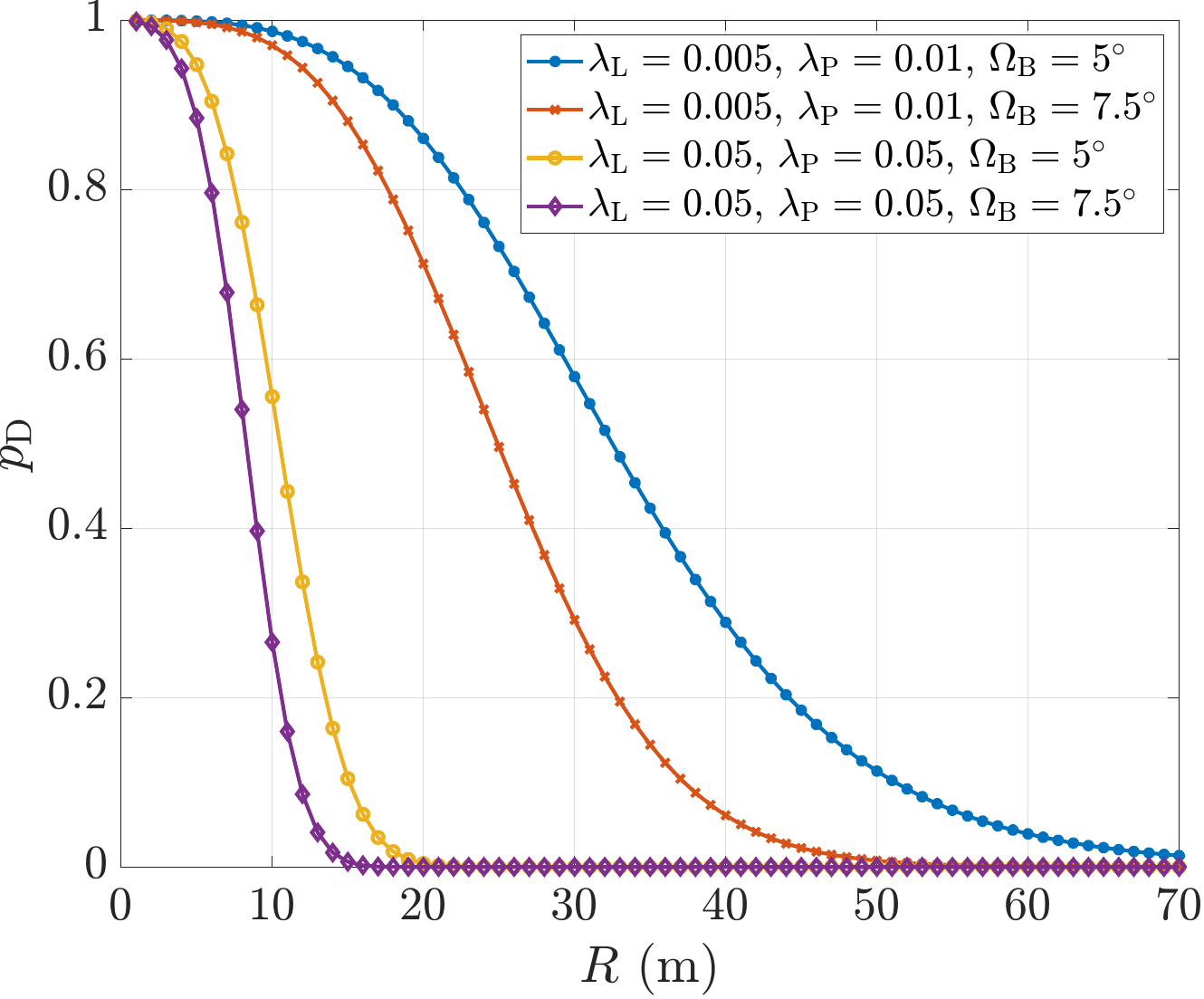}
\label{fig:result_2}}
\hfil
\subfloat[]
{\includegraphics[width=0.32\textwidth]{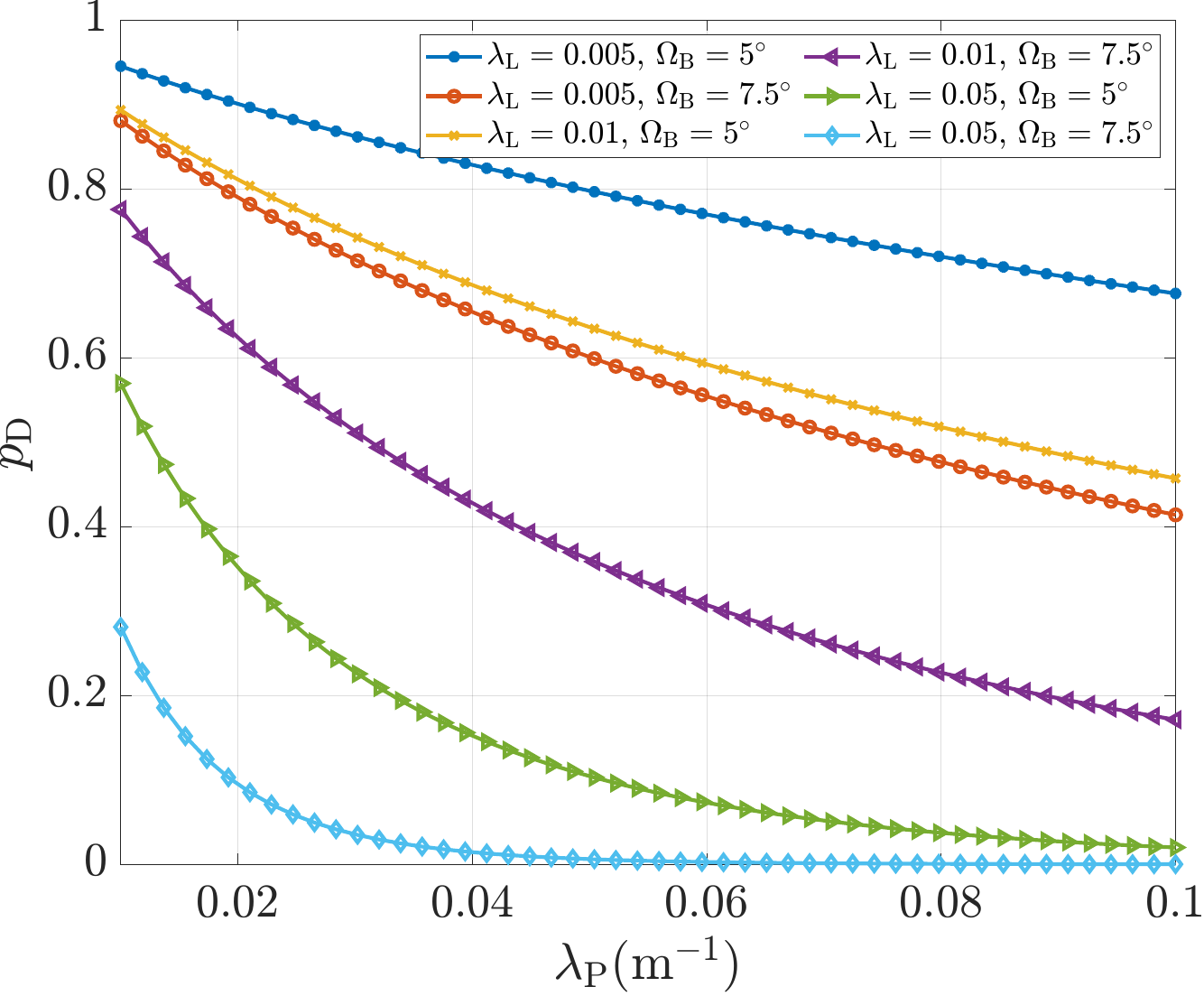}
\label{fig:result_4}}
\hfil
\subfloat[]
{\includegraphics[width=0.32\textwidth]{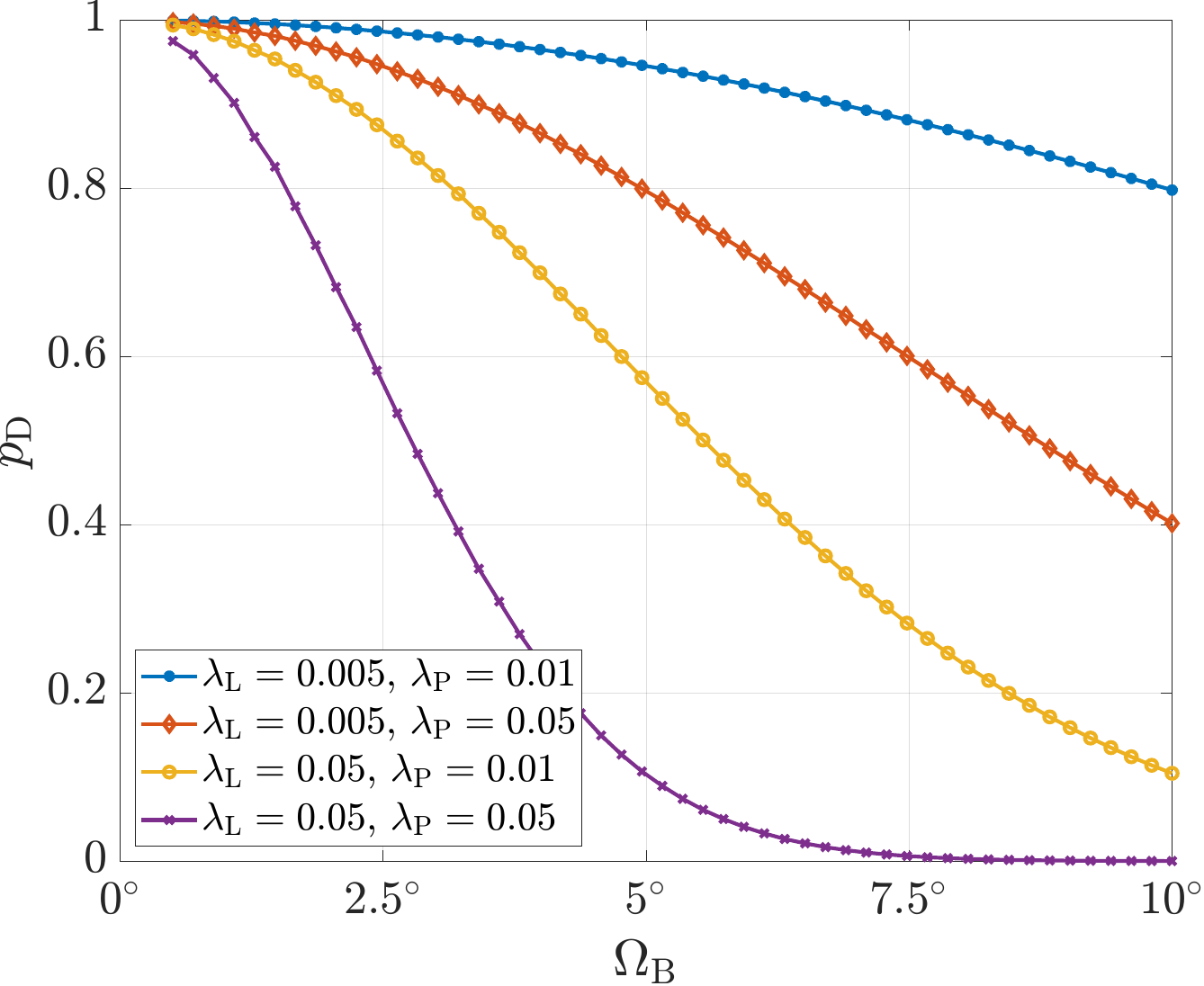}
\label{fig:result_5}}
\caption{Probability of successful detection $p_{\rm D}$ with respect to (a) $R$, (b) $\lambda_{\rm P}$, and (c) $2\Omega_{\rm B}$. (d) $n_{\rm D}$ with respect to $2\Omega_{\rm B}$.}
\label{fig:result_1_6} 
\end{figure*}

\section{Numerical Results and Discussion}
\label{sec:Results}
This section highlights some of the results from a typical automotive radar \ac{PLCP} framework. The radar parameters used for the numerical results taken from~\cite{series2014systems} are as, 
% listed in table~\ref{tab:values}. 
$P = 10$ dBm, $\bar{\sigma} = 30$ dBsm, $\alpha = 2$, $G_{\rm t} = G_{\rm r} = 10$ dBi, $f_c = 76.5$ GHz, $N_{\rm d} = -174$ dBm/Hz, $W = 25$ KHz, and $\beta = 10$ dB.
To better understand the performance of such a line process-based radar framework, we will see how the success probability $p_{\rm D}$ varies with system parameters like $R$, $\Omega_{\rm B}$, $\lambda_{\rm P}$ e.t.c. In all plots, $R=15 {\rm m}$ unless specified otherwise and the beamwidth is varied between $1^\circ$ and $20^\circ$. 

\subsection{Success Probability}
\label{sec:section_5}
% \begin{table}[h]
% \small
% \centering
% \begin{tabular}{|l|l|l|}
% \hline
% \textbf{Parameters}          & \textbf{Symbol}                 & \textbf{Value}       \\ \hline
% Transmit Power      & $P$                    & 10 dBm~\cite{series2014systems}      \\ \hline
% Mean RCS            & $\bar{\sigma}$         & 30 dBsm~\cite{series2014systems}     \\ \hline
% Path loss exponent  & $\alpha$               & 2           \\ \hline
% Antenna Gain        & $G_{\rm t} = G_{\rm r}$ & 10 dBi~\cite{series2014systems}            \\ \hline
% Center frequency    & $f_c$                  & 76.5 GHz~\cite{series2014systems}    \\ \hline
% Noise power density & $N_{\rm d}$            & -174 dBm/Hz~\cite{series2014systems} \\ \hline
% Bandwidth           & $W$               & 25 KHz~\cite{series2014systems}      \\ \hline
% SINR threshold           & $\beta$               & 10 dB~\cite{series2014systems}      \\ \hline
% \end{tabular}
% \caption{Summary of parameter values.}
% % \vspace*{0.15cm}
% \label{tab:values}
% \end{table}

We first plot $p_{\rm D}$ with respect to $R$, i.e., the ranging distance in Fig~\ref{fig:result_2}. The probability of successful detection decreases as $R$ increases irrespective of the value of beamwidth, $\lambda_{\rm L}$, and $\lambda_{\rm P}$. For smaller intensity values of streets and radars, the interference power received is smaller as compared to larger values of $\lambda_{\rm L}$ and $\lambda_{\rm P}$. Thus, we see that $p_{\rm D}$ is higher for smaller values of $\lambda_{\rm L}$ and $\lambda_{\rm P}$ due to less likelihood of strong interference from neighboring radars.

Next, we see in Fig~\ref{fig:result_4}, and~\ref{fig:result_5}, the success probability decreases with respect to $\lambda_{\rm P}$, and beamwidth respectively. $p_{\rm D}$ decreases as the two parameters increase because of the increase in the total number of radar interferers and the decrease of the distance between the nearest interfering radar and transmitting ego beam at the origin. In Fig~\ref{fig:result_4}, for different values of $\lambda_{\rm L}$ and beamwidth we see $p_{\rm D}$ decreases with $\lambda_{\rm P}$. For fixed value of beamwidth, $p_{\rm D}$ decreases as $\lambda_{\rm L}$ increases from $0.005$ to $0.05$ per ${\rm m}^{2}$. The effect of intensity on $p_{\rm D}$ showcases that by increasing the intensity of either vehicles and streets, we increase the number of interferers by a larger factor as compared to only an increase in the values of $\Omega_{\rm B}$. In Fig~\ref{fig:result_5} we now illustrate the effect of beamwidth on $p_{\rm D}$. We see that as the value of beamwidth increases, $p_{\rm D}$ decreases non-linearly. It showcases that having fewer interfering points results in higher successful detection probability, i.e., $\lambda_{\rm L} = 0.005 \,{\rm m}^{-2}$ and $\lambda_{\rm P} = 0.01 \,{\rm m}^{-1}$, as compared to when $\lambda_{\rm L} = 0.01 \,{\rm m}^{-2}$ and $\lambda_{\rm P} = 0.1 {\rm m}^{-1}$. 

Fig~\ref{fig:result_6} shows the plot of $n_{\rm D}$ with respect to beamwidth. From the plot, it is evident that an increase in the intensity of lines and the beamwidth leads to an increase in the average length of line segments falling in the radar sector. However, with the increase in $\lambda_{\rm L}$ and beamwidth, the probability of successful detection decreases due to increased interference. Consequently, a trade-off exists between $p_{\rm D}$ and the average length, as illustrated by Fig~\ref{fig:result_6} that initially exhibits an upward trend followed by a subsequent decline. Thus, there exists an optimal value of beamwidth, for which $n_{\rm D}$ is maximized.
% , indicating that we achieve an optimum number of successful detections.

\subsection{Optimal Parameter}
In Fig~\ref{fig:result_6}, we saw that there is the optimal value of beamwidth for which $n_{\rm D}$ is maximized i.e     $\Omega_{\rm B}^\ast= \arg\max n_{\rm D}$. In Fig~\ref{fig:result_7} we plot the optimal beamwidth $\Omega_{\rm B}^\ast$ w.r.t $\lambda_{\rm P}$ for 3 different values of $\lambda_{\rm L}$. The figure illustrates that the optimal beamwidth decreases with an increase in $\lambda_{\rm P}$, eventually saturating. In our experiments, as mentioned earlier, the value of beamwidth is varied from $1^\circ$ to $20^\circ$, and we see that optimum values for each value of $\lambda_{\rm P}$ lies in between these values of beamwidth. We also see that optimal beamwidth decreases as $\lambda_{\rm L}$ increases. Therefore, we conclude that the optimal beamwidth is lower when there is a higher density of interfering radar beams.

Likewise, we find the optimal beamwidth as a function of  $R$. Fig~\ref{fig:result_8} illustrates the ideal beamwidth with respect to $R$ for various values of $\lambda_{\rm L}$ and $\lambda_{\rm P}$. We see that for smaller values of $R$, irrespective of the intensity of streets and radars, the optimal beamwidth remains constant at $20^\circ$. For smaller values of $\lambda_{\rm L}$ and $\lambda_{\rm P}$, optimal beamwidth remains constant for a larger range as compared to higher values of $\lambda_{\rm L}$ and $\lambda_{\rm P}$. In Fig.~\ref{fig:result_6}, we saw that for $\lambda_{\rm L} = 0.01 \,{\rm m}^{-2}$ and $\lambda_{\rm P} = 0.01 \,{\rm m}^{-1}$, $n_{\rm D}$ increases with increase in beamwidth, and does not decrease for the given range of beamwidth. Our analysis shows that the number of successful detections keeps increasing for smaller values of ranging distance $R$ with the increasing beamwidth. This is because with very few interferers present and a small $R$, even if we increase beamwidth, the effect of interferers will not be prominent enough to cause the $n_{\rm D}$ to decrease. This is the same reason when $\lambda_{\rm L} = 0.01 \,{\rm m}^{-2}$ and $\lambda_{\rm P} = 0.01 \,{\rm m}^{-1}$, optimal beamwidth is constant at $20^\circ$ for a larger range of $R$, as compared to $\lambda_{\rm L} = 0.05 \,{\rm m}^{-2}$ and $\lambda_{\rm P} = 0.05 \,{\rm m}^{-1}$. Eventually, optimal beamwidth saturates to some value, but we see that for higher intensity values of interfering radars, optimal beamwidth saturates quickly compared to lower values. By increasing the area of the radar sector with the increase in $R$, the effect of interfering signal power received for higher $\lambda_{\rm L}$ and $\lambda_{\rm P}$ causes the optimal beam to saturate to smaller value quickly as compared when $\lambda_{\rm L}$ and $\lambda_{\rm P}$ are smaller. 

% The findings of this study hold significant ramifications for practical implementations, specifically. We can enhance energy efficiency, reduce costs, and increase sustainability-based radar systems by identifying the effective beamwidth for automotive radar systems. Additionally, the mitigation of interference and enhancements in radar performance can enhance collision avoidance, lane-keeping assistance, and pedestrian identification capabilities. Therefore, there is a potential for a substantial decrease in traffic accidents and fatalities. 

\begin{figure*}[t]
\centering
\subfloat[]
{\includegraphics[width=0.32\textwidth]{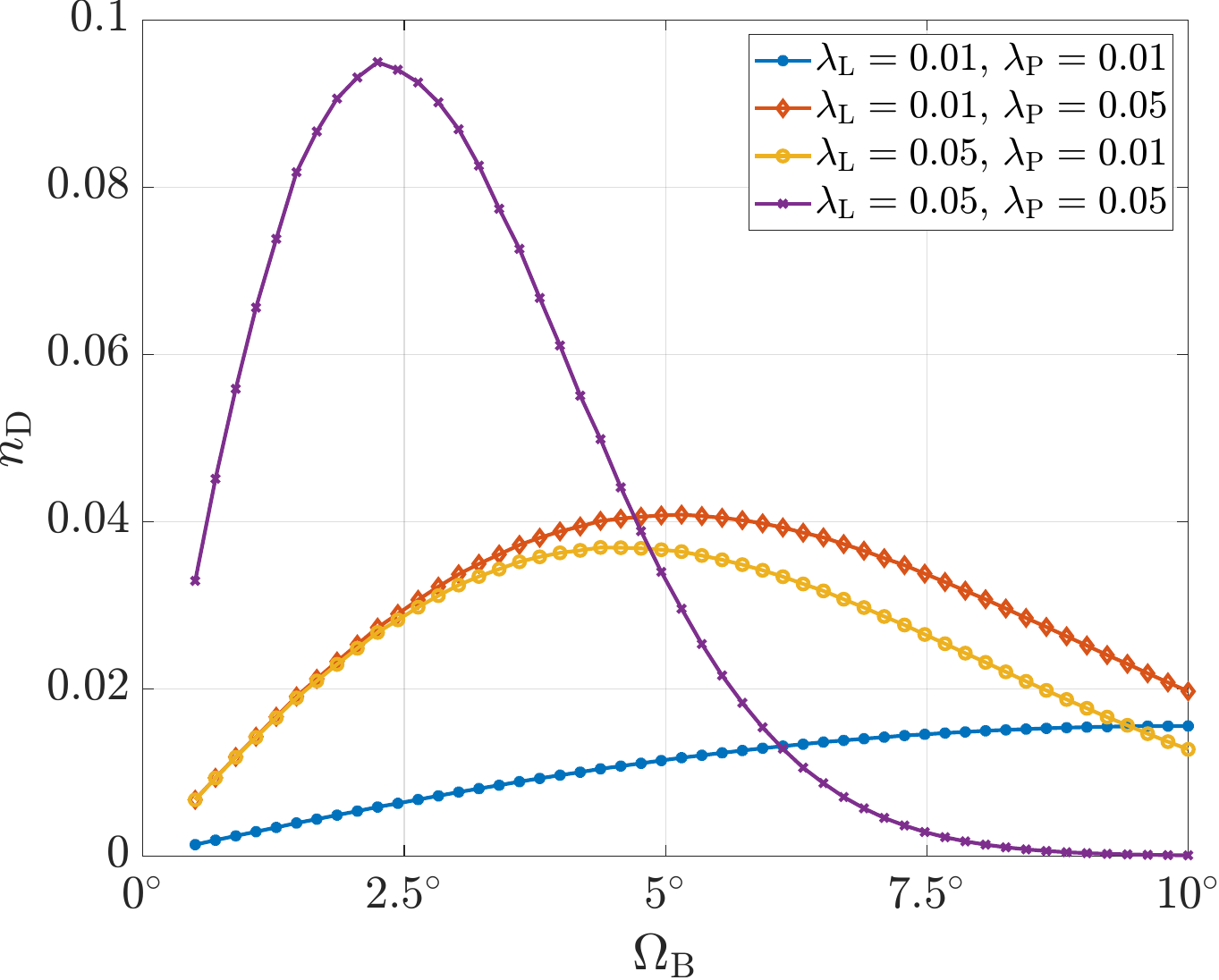}
\label{fig:result_6}}
\hfil
\subfloat[]
{\includegraphics[width=0.32\textwidth]{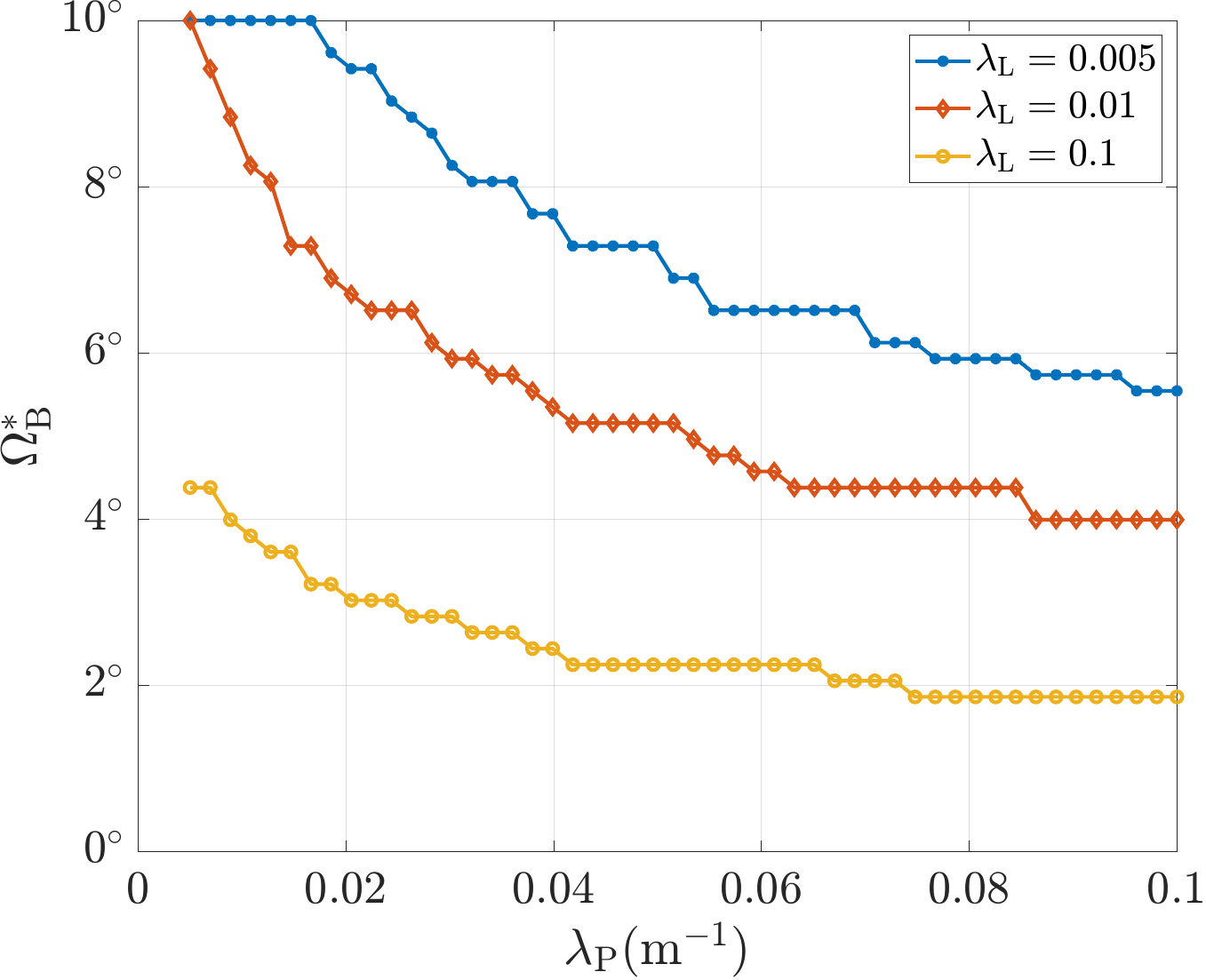}
\label{fig:result_7}}
\hfil
\subfloat[]
{\includegraphics[width=0.32\textwidth]{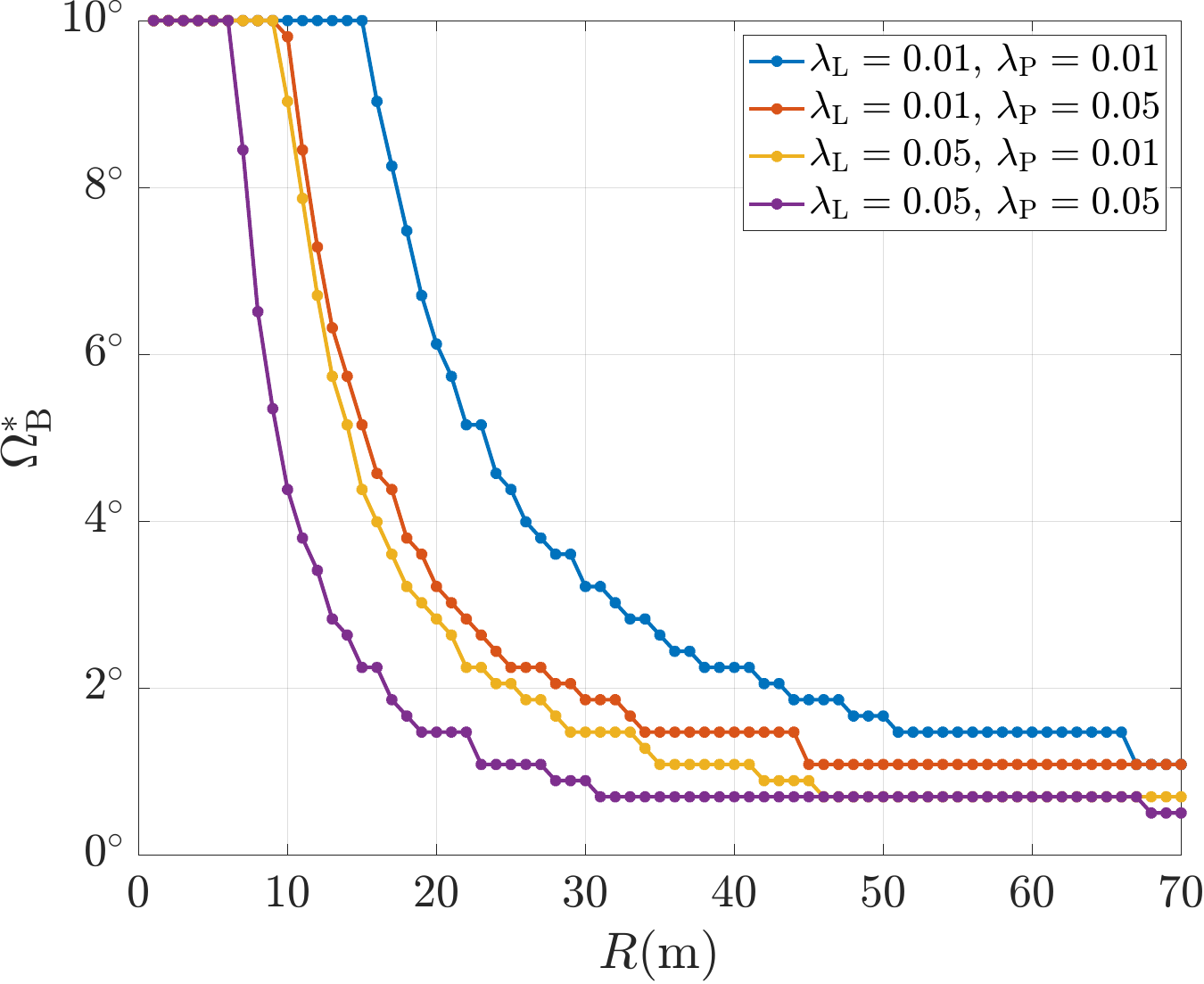}
\label{fig:result_8}}
\caption{Optimal beamwidth with respect to (a) $\lambda_{\rm P}$, and (b) $R$}
\label{fig:result_7_8} 
\end{figure*}

\section{Conclusion}
\label{sec:Conclusion}
This study presents an innovative methodology for modeling vehicular radar systems, considering the spatial arrangement of road networks in urban settings. Applying a Poisson line Cox process (PLCP) model enables the examination of interference attributes, resulting in the formulation of significant findings for system design. Beamwidth optimization is considered a critical component that serves as a basis for improving the efficiency and effectiveness of vehicles equipped with radar systems. The findings suggest that there is a trade-off between the width of the radar beam and the number of successful detections. %These results have practical implications for enhancing energy efficiency and reducing costs in radar systems. 
The results of this investigation offer prospective applications in collision avoidance, lane-keeping assistance, and pedestrian identification. 
% Ultimately, these advancements have the potential to contribute to the mitigation of traffic accidents and fatalities.

\bibliographystyle{ieeetr}
\bibliography{references}

\end{document}